\documentclass[12pt,a4]{article}

\usepackage{color,tikz}
\usepackage[unicode,bookmarks,bookmarksopen,bookmarksopenlevel=2,colorlinks,linkcolor=blue,citecolor=green]{hyperref}

\usepackage{amsmath,eucal,amssymb,amsthm,amsfonts}
\usepackage{mathrsfs,graphicx,texdraw,enumerate}

\newtheorem{remark}{Remark}
\DeclareMathOperator{\tr}{tr}
\DeclareMathOperator{\rank}{rank}
\DeclareMathOperator{\End}{End}
\DeclareMathOperator{\str}{str}
\DeclareMathOperator{\sdet}{sdet}
\DeclareMathOperator{\diag}{diag}

\newcommand{\field}[1]{\mathbb{#1}}

\numberwithin{equation}{section}

\def\openone{\leavevmode\hbox{\small1\kern-3.3pt\normalsize1}}

\def\diag{\mbox{diag}\,}
\def\tr{\mbox{tr}\,}

\newtheorem{proposition}{Proposition}

\begin{document}

\begin{center}
{\LARGE \bf Grassmann extensions of Yang-Baxter maps}

\bigskip

{\bf G. G. Grahovski$^{a,b,}$\footnote{E-mail: {\tt grah@essex.ac.uk}}, S. Konstantinou-Rizos$^{c,d,}$\footnote{E-mail: {\tt skonstantin@chesu.ru}} and  A. V. Mikhailov$^{d,}$\footnote{E-mail: {\tt A.V.Mikhailov@leeds.ac.uk}}}

\end{center}

\medskip

\noindent
{\it $^a$ Department of Mathematical Sciences, University of Essex,
CO4 3SQ Colchester,  UK}\\
{\it $^b$ Institute of Nuclear Research and Nuclear Energy,  Bulgarian Academy of Sciences,\\ 72 Tsarigradsko chausee, Sofia 1784, Bulgaria }\\
{\it $^c$ Institute of Mathematical Physics \& Seismodynamics, Chechen State University,\\ 33 Kievskaya str, 364907 Grozny, Russia }\\
{\it $^d$ Department of Applied Mathematics, University of Leeds,  Leeds, LS2 9JT, UK}\\

\begin{abstract}
\noindent
In this paper we show that there are explicit Yang-Baxter maps with Darboux-Lax representation between Grassmann extensions of algebraic varieties. Motivated by some recent results on noncommutative extensions of Darboux transformations, we first derive a Darboux matrix associated with the Grassmann-extended derivative Nonlinear Schr\"odinger (DNLS) equation, and then we deduce novel endomorphisms of Grassmann varieties, which possess the Yang-Baxter property. In particular, we present ten-dimensional maps which can be restricted to eight-dimensional Yang-Baxter maps on invariant leaves, related to the Grassmann-extended NLS and DNLS equations. We consider their vector generalisations.
\end{abstract}

\section{Introduction}
The Yang-Baxter (YB) equation has a fundamental role in the theory of quantum and classical integrable systems. In particular, the set-theoretical solutions of the YB equation, have been of great interest for several researchers in the area of Mathematical Physics. The consideration of such solutions was formally proposed by Drinfeld in \cite{Drinfel'd}, although the first examples of such solutions appeared in \cite{Sklyanin}. Moreover, the study of the set-theoretical solutions gained a more algebraic flavour in \cite{Buchstaber}. We refer to these solutions using the shorter term ``Yang-Baxter maps'' which was proposed by Veselov in \cite{Veselov}. YB maps are related to several concepts of integrability as, for instance, the multidimensionally consistent equations \cite{ABS-2004, ABS-2005, Bobenko-Suris,Frank4}. Of particular interest are those Yang-Baxter maps which admit Lax representation \cite{Veselov2}. They are connected with integrable mappings \cite{Veselov, Veselov3} and they are also related to integrable partial differential equations via Darboux transformations \cite{Sokor-Sasha}.

Moreover, noncommutative extensions of integrable equations have been of great interest over the last decades \cite{Chain-Kulish, Dimakis}. Darboux transformations for noncommutative-extended integrable equations were recently constructed; in the case of Grassman-extended NLS equation in \cite{Georgi-Sasha} and for the supersymmetric KdV  equation \cite{Liu,Liu2,Zhou} and the AKNS system \cite{Liu3}. At the same time, the derivation of noncommutative versions of YB maps has gained its interest \cite{Doliwa}.

In this paper, we make the first attempt to extend the theory of YB maps in the case of Grassmann algebras; in particular, we study the Grassmann extensions of the YB maps related to the Nonlinear Schr\"odinger equation and the derivative Nonlinear Schr\"odinger equation which have recently appeared in \cite{Sokor-Sasha}.

The paper is organised as follows. The second section deals with parametric YB maps and their Lax representations. Moreover, we present some basic properties of Grassmann algebras in order to make this text self-contained, as well as some properties of YB maps which admit Lax representation. In section 3, following \cite{Georgi-Sasha}, we consider a noncommutative (Grassmann) extension of the Darboux transformation for the DNLS equation. In section 4, we employ all the Darboux matrices presented in section 3 and, from their associated refactorisation problems, we construct ten-dimensional YB maps. The entries of the considered Darboux matrices satisfy particular systems of differential-difference equations which possess first integrals. These integrals indicate that the associated YB maps can be restricted to eight-dimensional YB maps on invariant leaves. Moreover, we consider their vector generalisations. Finally, in section 5 we summarise our results and we present some ideas for future work.

\section{Preliminaries}
Let $A$ be an algebraic variety in $K^N$, where $K$ is any field of zero characteristic (such as $\field{C}$, $\field{R}$  or $\field{Q}$), and let $Y\in \End(A\times A)$ be a map $(x,y)\stackrel{Y}{\mapsto}(u(x,y),v(x,y))$. The map $Y$ is called a \textit{Yang-Baxter map} if it satisfies the following \textit{Yang-Baxter equation}
\begin{equation}\label{YB_eq1}
Y^{12}\circ Y^{13} \circ Y^{23}=Y^{23}\circ Y^{13} \circ Y^{12},
\end{equation}
where $Y^{ij}\in \End(A\times A\times A)$, $i,j=1,2,3$, $i\neq j$, are defined by the following relations
\begin{subequations}\label{Yij}
\begin{align}
Y^{12}(x,y,z)&=(u(x,y),v(x,y),z), \\
Y^{13}(x,y,z)&=(u(x,z),y,v(x,z)), \\
Y^{23}(x,y,z)&=(x,u(y,z),v(y,z)),
\end{align}
\end{subequations}
where $x,y,z\in A$.

A YB map $Y$ is called \textit{reversible} if the composition of $\tilde{Y}=\pi Y\pi$ (where $\pi\in\End(A\times A)$ is the permutation map $\pi(x,y)=(y,x)$) with $Y$ is the identity map, namely
\begin{equation}\label{reversible}
\tilde{Y}\circ Y=Id.
\end{equation}
Furthermore, we use the term \textit{parametric YB map} if two parameters $a,b\in K$ are involved in the definition of the YB map, namely we have a map of the following form
\begin{equation}
Y_{a,b}:(x,y)\mapsto (u(x,y;a,b),v(x,y;a,b)),
\end{equation}
satisfying the \textit{parametric YB equation}
\begin{equation}\label{YB_eq}
Y^{12}_{a,b}\circ Y^{13}_{a,c} \circ Y^{23}_{b,c}=Y^{23}_{b,c}\circ Y^{13}_{a,c} \circ Y^{12}_{a,b}.
\end{equation}

\subsection{Grassmann extended varieties}
Here, we briefly present the basic properties of Grassmann algebras. For further details one could consult \cite{Berezin}. Let $G$ be a $\field{Z}_2$-graded algebra over $\field{C}$ or, in general, over a field $K$ of characteristic zero. Thus, as a linear space $G$  is a direct sum $G=G_0\oplus G_1$ (mod 2), such that $G_iG_j\subseteq G_{i+j}$. Those elements of $G$ that belong either to $G_0$ or to $G_1$ are called \textit{homogeneous}, the ones from $G_0$ are called \textit{even} (bosonic), while those in $G_1$ are called \textit{odd} (fermionic).

By definition, the parity $|a|$ of an even homogeneous element $a$ is $0$, and it is $1$ for odd homogeneous elements. The parity of the product $|ab|$ of two homogeneous elements is a sum of their parities: $|ab|=|a|+|b|$. Grassmann commutativity means that $ba=(-1)^{|a||b|}ab$ for any homogeneous elements $a$ and $b$. In particular, $\alpha^2=0$, for all odd elements $\alpha\in G_1$, and even elements commute with all the elements of $G$.

\begin{remark}\normalfont
In the rest of this paper we shall be using Latin letters for even variables, and Greek letters when referring to the odd ones; yet, we shall be using the Greek letter $\lambda$ when referring to the spectral parameter, despite the fact that $\lambda$ is an even variable.
\end{remark}

A \textit{Grassmann extension of an algebraic variety}, $V_G(p_1,\ldots,p_k)$, can be defined similarly to the commutative case:
{\small \begin{equation}
V_G(p_1,\ldots,p_k):=\{a_1,\ldots,a_n\in G_0, \alpha_1,\ldots,\alpha_m\in G_1|p_1=\dots p_k=0,p_i\in \field{C}\left[a_1,\ldots,a_n,\alpha_1,\ldots,\alpha_m\right]\}
\end{equation}}

\subsubsection{Supertrace and superdeterminant}
Let $M$ be a square matrix of the following form
\begin{equation}\label{G-matrixForm}
M=\left(
\begin{matrix}
 P & \Pi \\
 \Lambda & L
\end{matrix}\right),
\end{equation}
where $P$ and $L$ are square matrices of even entries,\index{even variable(s)} whereas $\Pi$ and $\Lambda$ are matrices with odd entries,\index{odd variable(s)} not necessarily square.

We define the \textit{supertrace} \index{supertrace} of $M$ --and we shall denote it by $\str(M)$-- to be the following quantity
\begin{equation}
\str(M)=\tr (P)-\tr (L),
\end{equation}
where $\tr(.)$ is the usual trace of a matrix.

Moreover, we define the \textit{superdeterminant}\index{superdeterminant} of $M$ --and we shall denote it by $\sdet(M)$-- to be
\begin{equation}
\sdet(M)=\det(P-\Pi L^{-1}\Lambda)\det(L^{-1})=\det(P^{-1})\det(L-\Lambda P^{-1}\Pi),
\end{equation}
where $\det(.)$ is the usual determinant of a matrix.

\subsection{Lax representations of YB maps}
Following Suris and Veselov in \cite{Veselov2}, we call a \textit{Lax matrix for a parametric YB map} a square matrix, $L=L(x,\chi;a,\lambda)$, depending on an even variable $x$, an odd variable $\chi$, a parameter $a$ and a \textit{spectral parameter} $\lambda$, such that the \textit{Lax-equation}
\begin{equation}\label{eqLax}
L(u,\xi;a)L(v,\eta;b)=L(y,\psi;b)L(x,\chi;a)
\end{equation}
is satisfied whenever $(u,\xi,v,\eta)=Y_{a,b}(x,\chi,y,\psi)$. Equation (\ref{eqLax}) is also called a \textit{refactorisation problem}.

If the Lax-equation (\ref{eqLax}) has a unique solution, namely it is equivalent to a map
\begin{equation}\label{unique-sol}
  (u,v,\xi,\eta)=Y_{a,b}(x,y,\chi,\psi),
\end{equation}
then the Lax matrix $L$ is said to be \textit{strong} \cite{Kouloukas2}. In this case (\ref{unique-sol}) is a Yang-Baxter map and it is reversible \cite{Veselov3}.

Now, since the Lax equation (\ref{eqLax}) has the obvious symmetry
\begin{equation}\label{symL}
(u,\,\xi,\,v,\,\eta,\,a,\,b) \longleftrightarrow  (y,\,\psi,\,x ,\,\chi,\,b,\,a)
\end{equation}
we have the following
\begin{proposition} \label{rationality}
If a matrix refactorisation problem (\ref{eqLax})\index{refactorisation problem(s)} yields a rational map $(x,\chi,y,\psi)=Y_{a,b}(u,\xi, v,\eta)$, then this map is birational.
\end{proposition}

\begin{proof} Let $Y:(x,\chi,y,\psi)\mapsto (u,\xi,v,\eta)$ be a rational map corresponding to a refactorisation problem\index{refactorisation problem(s)} (\ref{eqLax}), i.e.
\begin{subequations}
\begin{align}
x\mapsto u&=\frac{n_1(x,\chi,y,\psi;a,b)}{d_1(x,\chi,y,\psi;a,b)}, \qquad y\mapsto v=\frac{n_2(x,\chi,y,\psi;a,b)}{d_2(x,\chi,y,\psi;a,b)},\\
\chi\mapsto \xi &=\frac{n_3(x,\chi,y,\psi;a,b)}{d_3(x,\chi,y,\psi;a,b)}, \qquad \psi\mapsto \eta=\frac{n_4(x,\chi,y,\psi;a,b)}{d_4(x,\chi,y,\psi;a,b)},
\end{align}
\end{subequations}
where $n_i$, $d_i$, $i=1,2,3,4$, are polynomial functions of their variables.

Due to the symmetry (\ref{symL}) of the refactorisation problem (\ref{eqLax}),\index{refactorisation problem(s)} the inverse map of $Y$, $Y^{-1}:(u,\xi,v,\eta)\mapsto (x,\chi,y,\psi)$, is also rational and it is given by
\begin{subequations}
\begin{align}
u\mapsto x&=\frac{n_1(v,\eta,u,\xi;b,a)}{d_1(v,\eta,u,\xi;b,a)}, \qquad v\mapsto y=\frac{n_2(v,\eta,u,\xi;b,a)}{d_2(v,\eta,u,\xi;b,a)},\\
\xi\mapsto \chi&=\frac{n_3(v,\eta,u,\xi;b,a)}{d_3(v,\eta,u,\xi;b,a)}, \qquad \eta\mapsto \psi=\frac{n_4(v,\eta,u,\xi;b,a)}{d_4(v,\eta,u,\xi;b,a)},
\end{align}
\end{subequations}
Therefore, $Y$ is a birational map.
\end{proof}

\begin{remark}\normalfont
Functions $d_i(x,\chi,y,\psi;a,b)$, $i=1,2,3,4$, must be non-nilpotent even-valued.
\end{remark}

\begin{proposition}\label{genInv}
If $L=L(x,\chi,a;\lambda)$ is a Lax matrix with corresponding YB map $Y:(x,\chi,y,\psi)\mapsto (u,\xi,v,\eta)$, then $\str(L(y,\psi,b;\lambda)L(x,\chi,a;\lambda))$ is a generating function of invariants of the YB map.
\end{proposition}

\begin{proof}
Since,
\begin{eqnarray}
\str(L(u,\xi,a;\lambda)L(v,\eta,b;\lambda))&\overset{(\ref{eqLax})}{=}&\str(L(y,\psi,b;\lambda)L(x,\chi,a;\lambda))\nonumber\\
&=& \str(L(x,\chi,a;\lambda)L(y,\psi,b;\lambda)),\label{strace}
\end{eqnarray}
and function $\str(L(x,\chi,a;\lambda)L(y,\psi,b;\lambda))$ can be written as $\str(L(x,\chi,a;\lambda)L(y,\psi,b;\lambda))=\\ \sum_k \lambda^k I_k(x,\chi,y,\psi;a,b)$, from (\ref{strace}) follows that
\begin{equation}
I_i(u,\xi,v,\eta;a,b)=I_i(x,\chi,y,\psi;a,b),
\end{equation}
which are invariants for $Y$.
\end{proof}

Matrix $L(x,\chi,a;\lambda)L(y,\psi,b;\lambda)$ is called the \textit{monodromy matrix}.

\begin{remark} \normalfont
The invariants of a YB map, $I_i(x,\chi,y,\psi;a,b)$, may not be functionally independent.
\end{remark}


\section{Grassmann extensions of Darboux transformations}
Let $L$ be a Lax operator of the following form
\begin{equation}
L(p,q,\theta,\phi;\lambda)=D_x+U(p,q,\phi,\theta;\lambda),
\end{equation}
where $U$ is a matrix depending on two even potentials, $p=p(x)$ and $q=q(x)$, two odd potentials, $\theta=\theta(x)$ and $\phi=\phi(x)$, a spectral parameter $\lambda$ and a variable $x$ implicitly through the potentials. In all our cases the dependence on the spectral parameter is polynomial.

By Darboux transformation we understand a map of the following form
\begin{equation}\label{defDarboux}
L\rightarrow \tilde{L}=MLM^{-1},
\end{equation}
where $\tilde{L}$ is $L$ updated with potentials $p_{10}=p_{10}(x)$, $q_{10}=q_{10}(x)$, $\theta_{10}=\theta_{10}(x)$ and $\phi_{10}=\phi_{10}(x)$, namely $\tilde{L}=L(p_{10},q_{10},\theta_{10}, \phi_{10};\lambda)$. The matrix $M$ in (\ref{defDarboux}) is called  \textit{Darboux matrix}. Here, we shall be assuming that the matrix $M$ has the same $\lambda$-dependence with $U$. Moreover, we define the \textit{rank} of a Darboux transformation to be the rank of the matrix which appears as coefficient of the highest power of the spectral parameter.

In this section we consider the Grassmann extensions of the Darboux matrices corresponding to the NLS equation (see \cite{Georgi-Sasha}) and the DNLS equation, which we shall use to construct YB maps.

\subsection{Nonlinear Schr\"odinger equation}
The Grassmann extension of the Darboux matrix for the NLS equation was constructed in \cite{Georgi-Sasha}. In particular, the following noncommutative extension\footnote{If one sets the odd variables equal to zero, the obtained operator corresponds to the spatial part of the Lax pair for the NLS equation.} of the NLS operator
\begin{subequations}\label{G-NLS-U}
\begin{align}
&\mathcal{L} := D_x + U(p,q,\psi,\phi,\zeta,\kappa;\lambda)=D_x +\lambda U^{1}+U^{0},
\intertext{was considered, where $U^{1}$ and $U^0$ are given by}
&U^1={\rm{diag}}(1,-1,0),\quad U^0=\left(\begin{array}{ccc} 0 & 2p & \theta \\ 2q & 0 & \zeta \\ \phi & \kappa & 0\end{array}\right),
\end{align}
\end{subequations}
where $p,q\in G_0$ and $\psi,\phi,\zeta,\kappa \in G_1$.

It was shown that all the Darboux transformations of rank 1 associated to this operator are described by the following matrix
\begin{equation}\label{DarbouxNLS}
M(p,q,\theta,\phi;c_1,c_2)=\left(
\begin{matrix}
 F+\lambda & p & \theta\\
 q_{10} & c_1 & 0\\
 \phi_{10} & 0 & c_2
\end{matrix}\right),
\end{equation}
where $c_1$ and $c_2$ can be either 1 or 0. In the case where $c_1=c_2=1$, the entries of $M(p,q,\theta,\phi;1,1)$ satisfy the following system of differential-difference equations
\begin{subequations}\label{DarbouxNLS-sys}
\begin{align}
F_x&=2(pq-p_{10}q_{10})+\theta \phi-\theta_{10}\phi_{10},\\
p_x&=2(Fp-p_{10})+\theta \zeta,\\
q_{10,x}&=2(q-q_{10}F)-\kappa_{10} \phi_{10},\\
\theta_x&=F\theta-\theta_{10}+p\kappa,\\
\phi_{10,x}&=\phi-\phi_{10}F-\zeta_{10}q_{10},
\end{align}
\end{subequations}
and the algebraic equations
\begin{subequations}
\begin{align}
\theta q_{10}&=(\mathcal{S}-1)\kappa,\\
\phi_{10}p&=(\mathcal{S}-1)\zeta.
\end{align}
\end{subequations}
Moreover, system (\ref{DarbouxNLS-sys}) admits the following first integral
\begin{equation}\label{1integralNLS}
\partial_x(F-pq_{10}-\phi_{10}\theta)=0,
\end{equation}
which implies that $\partial_x(\sdet(M))=0$, since $\sdet(M)=\lambda+F-pq_{10}-\phi_{10}\theta$.

\subsection{Derivative nonlinear Schr\"odinger equation}
Let us now consider the Lax operator given by
\begin{subequations} \label{G-SL2-Lax-Op}
\begin{align}
&\mathcal{L}=D_x+\lambda^{2} U^2+\lambda U^1,
\intertext{where $U_1$ and $U_2$ are given by}
U^2=&\mathfrak{s}_3, \quad \quad U^1=\left(\begin{matrix}
 0 & 2p & 2\theta\\
 2q & 0 & 0\\
 2\phi & 0 & 0
\end{matrix}\right),
\end{align}
\end{subequations}
and $\mathfrak{s}_3=\diag(1,-1,-1)$. The potentials $p$ and $q$ and the spectral parameter $\lambda$ are even, whereas the potentials $\phi$ and $\theta$ are odd. Moreover, the operator (\ref{G-SL2-Lax-Op}) is invariant under the transformation
\begin{equation}\label{reduction-Z2-G}
s_1(\lambda): \mathcal{L}(\lambda) \rightarrow \mathcal{L}(-\lambda)=\mathfrak{s}_{3}\mathcal{L}(\lambda) \mathfrak{s}_{3}.
\end{equation}
We seek a rank 1 Darboux matrix of the following form
\begin{equation}\label{sqformM}
M=\lambda^2 M_2+\lambda M_1+M_0,
\end{equation}
where $M_i$, $i=0,1,2$, is a $3\times 3$ matrix, and we assume that $M$ possesses the same symmetry, (\ref{reduction-Z2-G}), as the Lax operator (\ref{G-SL2-Lax-Op}), namely
\begin{equation}\label{symmM}
M(-\lambda)=\mathfrak{s}_3 M(\lambda) \mathfrak{s}_3,
\end{equation}
as in the commutative case in \cite{SPS}. Therefore, for the entries of matrices $M_i$, $i=0,1,2$, we have
\begin{subequations}\label{form-M-Z2-G}
\begin{align}
M_{i,12}&=M_{i,13}=M_{i,21}=M_{i,31}=0,\quad i=0,2,\quad\text{and}\\
M_{1,11}&=M_{1,22}=M_{1,33}=M_{1,23}=M_{1,32}=0.
\end{align}
\end{subequations}
Now, the definition (\ref{defDarboux}) implies a second order algebraic equation in $\lambda$. Equating the coefficients of different powers of $\lambda$   to zero, we obtain the following system of equations
\begin{subequations}\label{Dar-eq}
\begin{align}
\left[U^2,M_2\right]&=0 \label{Dar-eq1} \\
\left[U^2,M_1\right]+U_{10}^1M_2-M_2U^1&=0 \label{Dar-eq2}\\
M_{2,x}+\left[U^2,M_0\right]+U_{10}^1M_1-M_1U^1&=0 \label{Dar-eq3}\\
M_{1,x}+U_{10}^1M_0-M_0U^1&=0 \label{Dar-eq4}\\
M_{0,x}&=0.\label{Dar-eq5}
\end{align}
\end{subequations}
Equation (\ref{Dar-eq1}) is  satisfied identically, whereas (\ref{Dar-eq5}) implies that matrix $M_0$ must be constant. Moreover, since $\rank{M_2=1}$, we can choose $M_2=\diag\{f,0,0\}$. In this case, from equation (\ref{Dar-eq2}) we have that the entries of $M_1$ are given by
\begin{equation}
M_{1,12}=fp,\quad M_{1,13}=f\theta,\quad M_{1,21}=q_{10}f \quad\text{and}\quad M_{1,31}=\phi_{10}f.
\end{equation}
Moreover, from equation (\ref{Dar-eq3}) we deduce equation
\begin{equation}\label{dif-difEqs-1}
f_x=2f(pq-p_{10}q_{10}+\theta\phi-\theta_{10}\phi_{10}),
\end{equation}
Therefore, matrix $M$ is of the form:
\small
\begin{equation}\label{DarbouxZ2-M}
M(f,p,q_{10},\theta,\phi_{10};c_1,c_2)=\lambda^2 \left(
\begin{matrix}
 f & 0 & 0\\
 0 & 0 & 0\\
 0 & 0 & 0
\end{matrix}\right)+\lambda \left(
\begin{matrix}
 0 & f p & f \theta \\
 q_{10} f & 0 &0 \\
 \phi_{10} f & 0 & 0
\end{matrix}\right)
+\left(
\begin{matrix}
 c_1 & 0 & 0 \\
 0 & c_2 &0 \\
 0 & 0 & 1
\end{matrix}\right),
\end{equation}
\normalsize
where we have chosen the constant matrix $M_0$ to be diagonal, namely of the form $M_0=\diag(c_1,c_2,1)$ (one of the parameters along its diagonal can be rescaled to 1).

Finally, due to (\ref{Dar-eq4}), the entries of the Darboux matrix must satisfy the following system of equations
\begin{subequations}\label{dif-difEqs}
\begin{align}
p_x&=2p(p_{10}q_{10}-pq+\theta_{10}\phi_{10}-\theta\phi)-2\frac{c_2p_{10}-c_1p}{f},\label{dif-difEqs-2}\\
q_{10,x}&=2q_{10}(p_{10}q_{10}-pq+\theta_{10}\phi_{10}-\theta\phi)-2\frac{c_1q_{10}-c_2q}{f},\\
\theta_x&=2\phi(p_{10}q_{10}-pq+\theta_{10}\phi_{10}-\theta\phi)+2\frac{c_1\theta-\theta_{10}}{f},\\
\phi_{10,x}&=2\phi_{10}(p_{10}q_{10}-pq+\theta_{10}\phi_{10}-\theta\phi)+2\frac{\phi-c_1\phi_{10}}{f}.\label{dif-difEqs-5}
\end{align}
\end{subequations}
where we have made use of (\ref{dif-difEqs-1}).

Thus, matrix $M$ given by (\ref{DarbouxZ2-M}) constitutes a Darboux matrix for the Lax operator (\ref{G-SL2-Lax-Op}), if its entries satisfy the system of equations system $\{(\ref{dif-difEqs-1}),(\ref{dif-difEqs})\}$. We can readily show that the latter system admits the following first integral
\begin{equation}\label{G-z2-1stInt}
\partial_x(c_2f-f^2(pq_{10}+c_2\theta\phi_{10}))=0,
\end{equation}
by straightforward calculation. Using the above first integral we can show that $\partial_x(\sdet M)=0$.

\begin{remark}\normalfont
The Darboux matrix (31) in \cite{SPS} constitutes the bosonic limit of (\ref{DarbouxZ2-M}).
\end{remark}


\section{Derivation Yang-Baxter maps}
In \cite{Sokor-Sasha} we considered the case of Darboux matrices associated with Lax operators of NLS type, which correspond to a recent classification of \textit{automorphic} Lie algebras \cite{BuryPhD,Bury-Sasha,Lombardo}. We used these Darboux matrices to construct six-dimensional YB maps together with their four-dimensional restrictions on invariant leaves.

In this paper, we are interested in the Grassmann extensions of these YB maps. In particular, we shall discuss the cases of the YB maps associated with the NLS equation \cite{ZS} and the DNLS equation \cite{Kaup-Newell}.


\subsection{NLS case}
According to (\ref{DarbouxNLS}) we define the following matrix
\begin{equation}\label{mLaxNLS}
M(\textbf{x};\lambda)=\left(
\begin{matrix}
 X+\lambda & x_1 & \chi_1\\
 x_2 & 1 &0 \\
 \chi_2 & 0 & 1
\end{matrix}\right), \qquad \textbf{x}:=(x_1,x_2,\chi_1,\chi_2,X),
\end{equation}
Then, we substitute $M$ to the Lax equation (\ref{eqLax}) which leads to a system of polynomial equations. The corresponding algebraic variety is a union of two ten-dimensional components. The first one is
obvious from the refactorisation problem, and it corresponds to the trivial Yang-Baxter map, while the second one corresponds to a non-trivial ten-dimensional Yang-Baxter map. In particular, we have the following.

\begin{proposition}
The matrix refactorisation problem
\begin{equation}\label{ref-eq-M}
M(\textbf{u};\lambda)M(\textbf{v};\lambda)=M(\textbf{y};\lambda)M(\textbf{x};\lambda),
\end{equation}
where $M=M(\textbf{x};\lambda)$ is given by (\ref{mLaxNLS}), yields the permutation map
\begin{equation}\label{trivialYB}
 \textbf{x}\mapsto \textbf{u}=\textbf{y}, \qquad \textbf{y}\mapsto \textbf{v}=\textbf{x}, \nonumber
\end{equation}
and the following ten-dimensional Yang-Baxter map
\begin{subequations}\label{NLS5D}
\begin{align}
 x_1\, \mapsto \, u_1&=y_1-\frac{X-x_1x_2-\chi_1\chi_2-Y+y_1y_2+\psi_1\psi_2}{1+x_1y_2+\chi_1\psi_2} x_1,\label{NLS5D-u1}\\
 x_2\, \mapsto \, u_2&=y_2,\\
 \chi_1\, \mapsto \, \xi_1&=\psi_1-\frac{X-x_1x_2-Y+y_1y_2+\psi_1\psi_2}{1+x_1y_2} \chi_1,\label{NLS5D-chi1}\\
 \chi_2\, \mapsto \, \xi_2&=\psi_2,\\
 X\, \mapsto \, U&=\frac{X-x_1x_2-\chi_1\chi_2+(x_1y_2+\chi_1\psi_2)Y+y_1y_2+\psi_1\psi_2}{1+x_1y_2+\chi_1\psi_2}, \\
 y_1\, \mapsto \, v_1&=x_1,\\
 y_2\, \mapsto \, v_2&=x_2+\frac{X-x_1x_2-\chi_1\chi_2-Y+y_1y_2+\psi_1\psi_2}{1+x_1y_2+\chi_1\psi_2} y_2,\label{NLS5D-v2}\\
 \psi_1\, \mapsto \, \eta_1&=\chi_1,\\
 \psi_2\, \mapsto \, \eta_2&=\chi_2+\frac{X-x_1x_2-\chi_1\chi_2-Y+y_1y_2}{1+x_1y_2} \psi_2,\label{NLS5D-eta2}\\
 Y\, \mapsto \, V&=\frac{(x_1y_2+\chi_1\psi_2)X+x_1x_2+\chi_1\chi_2+Y-y_1y_2-\psi_1\psi_2}{1+x_1y_2+\chi_1\psi_2},
\end{align}
\end{subequations}
which is non-involutive and birational.
\end{proposition}
\begin{proof}
Equation (\ref{ref-eq-M}) implies that $v_1=x_1$, $\eta_1=\chi_1$, $u_2=y_2$, $\xi_2=\psi_2$, and the following system of equations
\begin{subequations}\label{eqs-ref-prob-M}
\begin{align}
U+V&=X+Y\label{eqs-ref-prob-M-a}\\
UV+u_1v_2+\xi_1\eta_2&=YX+y_1x_2+\psi_1\chi_2\label{eqs-ref-prob-M-b}\\
Ux_1+u_1=Yx_1+y_1,& \quad U\chi_1+\xi_1=Y\chi_1+\psi_1\label{eqs-ref-prob-M-c}\\
y_2V+v_2=y_2 X+x_2&, \quad \psi_2V+\eta_2=\psi_2X+\chi_2\label{eqs-ref-prob-M-d},
\end{align}
\end{subequations}
for $u_1, \xi_1, U, v_2, \eta_2$ and $V$. From (\ref{eqs-ref-prob-M-c})-(\ref{eqs-ref-prob-M-d}) using (\ref{eqs-ref-prob-M-a}), we can express all variables $u_1, \xi_1, v_2$, and $\eta_2$ in terms of $Y-U$, as:
\begin{subequations}\label{uxiveta}
\begin{align}
u_1&=(Y-U)x_1+y_1, \qquad \xi_1=(Y-U)\chi_1+\psi_1\label{uxiveta-a}\\
v_2&=x_2-y_2(Y-U), \qquad \eta_2=\chi_2-\psi_2(Y-U)\label{uxiveta-b}.
\end{align}
\end{subequations}
Now, substituting (\ref{uxiveta}) to (\ref{eqs-ref-prob-M-b}), using (\ref{eqs-ref-prob-M-a}), we obtain
\begin{equation}
(Y-U)\left[U(1+x_1y_2+\chi_1\psi_2)-X+x_1x_2+\chi_1\chi_2-(x_1y_2+\chi_1\psi_2)Y-y_1y_2-\psi_1\psi_2 \right]=0.
\end{equation}

From the above follows that either $U=Y$, which in view of (\ref{uxiveta}) and (\ref{eqs-ref-prob-M-a}) implies the permutation map (\ref{trivialYB}), or
\begin{equation}
U=\frac{X-x_1x_2-\chi_1\chi_2+(x_1y_2+\chi_1\psi_2)Y+y_1y_2+\psi_1\psi_2}{1+x_1y_2+\chi_1\psi_2}.
\end{equation}
Substitution of the latter to (\ref{uxiveta}), implies that $u_1$ and $v_2$ are given by (\ref{NLS5D-u1}) and (\ref{NLS5D-v2}), respectively, while $\xi_1$ and $\eta_2$ are given by
\begin{eqnarray*}
\xi_1&=&\psi_1-\frac{X-x_1x_2-\chi_1\chi_2-Y+y_1y_2+\psi_1\psi_2}{1+x_1y_2+\chi_1\psi_2}\chi_1\\
\eta_2&=&\chi_2+\frac{X-x_1x_2-\chi_1\chi_2-Y+y_1y_2+\psi_1\psi_2}{1+x_1y_2+\chi_1\psi_2}\psi_2.
\end{eqnarray*}
Now, in the above expressions we mutliply both the nominator and the denominator with the conjugate expression of the latter, namely ``$1+x_1y_2-\chi_1\psi_2$", and we use the fact that $\chi_1^2=\psi_2^2=0$. Then, $\xi_1$ and $\eta_2$ can be written in the form (\ref{NLS5D-chi1}) and (\ref{NLS5D-eta2}).

Finally, it can be readily verified by straightforward calculation that map (\ref{NLS5D}) is non-involutive, and its birationality is due to Prop. \ref{rationality}.
\end{proof}

\begin{remark}\normalfont
The bosonic limit of the above map (namely if we set the odd variables $\chi_1=\chi_2=\psi_1=\psi_2=0$) is map (4.7) in \cite{Sokor-Sasha}.
\end{remark}

\subsubsection{Restriction on invariant leaves: Extension of Adler-Yamilov map}
In this section, we derive an eight-dimensional Yang-Baxter map from map (\ref{NLS5D}), which is the Grassmann extension of the Adler-Yamilov map \cite{Adler-Yamilov, Kouloukas, PT}. Our proof is motivated by the existence of the first integral (\ref{1integralNLS}) for system (\ref{DarbouxNLS-sys}).

In particular, we have the following.

\begin{proposition}
\begin{enumerate}
	\item The quantities $\Phi =X-x_1x_2-\chi_1\chi_2$ and $\Psi=Y-y_1y_2-\psi_1\psi_2$ are invariants (first integrals) of the map (\ref{NLS5D}).
	\item The ten-dimensional map (\ref{NLS5D}) can be restricted to an eight-dimensional map, $Y_{a,b}\in\End\{A_a\times A_b\}$, given by
	\begin{subequations}\label{Adler-Yamilov-Grassmann}
\begin{align}
 \textbf{x}\, \mapsto \, \textbf{u}&=\left(y_1+\frac{(b-a)(1+x_1y_2-\chi_1\psi_2)}{(1+x_1 y_2)^2}x_1,\,\,y_2,\,\,\psi_1+\frac{b-a}{1+x_1y_2}\chi_1,\,\,\psi_2\right),  \\
 \textbf{y}\, \mapsto \, \textbf{v}&=\left(x_1,\,\,x_2+\frac{(a-b)(1+x_1y_2-\chi_1\psi_2)}{(1+x_1y_2)^2}y_2,\,\,\chi_1,\,\,\chi_2+\frac{a-b}{1+x_1y_2}\psi_2 \right),
\end{align}
\end{subequations}
	where $a,b\in G_0$ and $A_a$, $A_b$ are level sets of the first integrals $\Phi$ and $\Psi$, namely
\begin{equation}\label{symleaves}
A_a=\{(x_1,x_2,\chi_1,\chi_2,X)\in A^5; \Phi=a\}, \quad A_b=\{(y_1,y_2,\psi_1,\psi_2,Y)\in A^5; \Psi=b\}.
\end{equation}
  \item The bosonic limit of map $Y_{a,b}$ is the Adler-Yamilov map.
\end{enumerate}
\end{proposition}

\begin{proof}
\begin{enumerate}
	\item It can be readily verified that (\ref{NLS5D}) implies $U-u_1u_2-\xi_1\xi_2=X-x_1x_2-\chi_1\chi_2$ and $V-v_1v_2-\eta_1\eta_2=Y-y_1y_2-\psi_1\psi_2$. Thus, $\Phi$ and $\Psi$ are invariants, i.e.  first integrals of the map.
	\item The existence of the restriction is obvious. Using the conditions $X=x_1x_2+\chi_1\chi_2+a$ and $Y=y_1y_2+\psi_1\psi_2+b$, one can eliminate $X$ and $Y$ from (\ref{NLS5D}). The resulting map, $\textbf{x}\rightarrow \textbf{u}(\textbf{x},\textbf{y})$, $\textbf{y}\rightarrow \textbf{v}(\textbf{x},\textbf{y})$, is given by (\ref{Adler-Yamilov-Grassmann}).
\item If one sets the odd variables of the above map equal to zero, namely $\chi_1=\chi_2=0$ and $\psi_1=\psi_2=0$, then the map (\ref{Adler-Yamilov-Grassmann}) coincides with the Adler-Yamilov map.
\end{enumerate}
\end{proof}

Now, one can use the condition $X=x_1x_2+\chi_1\chi_2+a$ to eliminate $X$ from the Lax matrix (\ref{mLaxNLS}), i.e.
\begin{equation} \label{laxSNLS}
M(\textbf{x};a,\lambda)=\left(
\begin{matrix}
 a+x_1x_2+\chi_1\chi_2+\lambda & x_1 & \chi_1\\
 x_2 & 1 &0 \\
 \chi_2 & 0 & 1
\end{matrix}\right),
\end{equation}
which corresponds to the Darboux matrix derived in \cite{Georgi-Sasha}.
Now, the Adler-Yamilov map's extension follows from the strong Lax representation
\begin{equation} \label{lax_eq_NLS}
  M(\textbf{u};a,\lambda)M(\textbf{v};b,\lambda)=M(\textbf{y};b,\lambda)M(\textbf{x};a,\lambda).
\end{equation}
Therefore, the extension of the Adler-Yamilov's map (\ref{Adler-Yamilov-Grassmann}) is a reversible parametric YB map. Moreover, it is easy to verify that it is not involutive. Birationality of map (\ref{Adler-Yamilov-Grassmann}) is due to Prop. \ref{rationality}.

To generate invariants of map (\ref{Adler-Yamilov-Grassmann}) we use $\str(M(\textbf{y};b,\lambda)M(\textbf{x};a,\lambda))$, and we obtain the following
\begin{eqnarray*}
&& T_1=x_1x_2+y_1y_2+\chi_1\chi_2+\psi_1\psi_2, \\
&& T_2=(a+x_1x_2+\chi_1\chi_2)(b+y_1y_2+\psi_1\psi_2)+x_1y_2+x_2y_1+\chi_1\psi_2-\chi_2\psi_1,
\end{eqnarray*}
where we have omitted the additive constants. However, $T_1$ and $T_2$ are linear combinations of the following invariants
\begin{subequations}\label{invariants}
\begin{align}
& I_1=x_1x_2+y_1y_2, \\
& I_2=\chi_1\chi_2+\psi_1\psi_2, \\
& I_3=\chi_1\chi_2\psi_1\psi_2,\\
& I_4=b(x_1x_2+\chi_1\chi_2)+a(y_1y_2+\psi_1\psi_2)+y_1y_2(x_1x_2+\chi_1\chi_2)+x_1x_2\psi_1\psi_2+\\
& ~\qquad     y_1x_2+y_2x_1+\chi_1\psi_2-\chi_2\psi_1.\nonumber
\end{align}
\end{subequations}

\subsection{DNLS case}
According to matrix $M(p,q_{10},\theta,\phi_{10};1,1)$ in (\ref{DarbouxZ2-M}) we consider the following matrix
\begin{equation} \label{DarbouxZ2}
M(\textbf{x};\lambda)=\lambda^2 \left(
\begin{matrix}
 X & 0 & 0\\
 0 & 0 & 0\\
 0 & 0 & 0
\end{matrix}\right)+\lambda \left(
\begin{matrix}
 0 & x_1 &  \chi_1 \\
 x_2 & 0 &0 \\
 \chi_2 & 0 & 0
\end{matrix}\right)
+\left(
\begin{matrix}
 1 & 0 & 0 \\
 0 & 1 &0 \\
 0 & 0 & 1
\end{matrix}\right),
\end{equation}
where $\textbf{x}=(x_1,x_2,\chi_1,\chi_2,X)$ and, in particular, we have set
\begin{equation}
X:=f,\quad x_1:=fp,\quad x_2=fq_{10},\quad \chi_1:=f\theta \quad \text{and}\quad \chi_2:=\psi_{10}f.
\end{equation}

In this case, the Lax equation implies the following equations
\begin{subequations}
\begin{eqnarray}
&U+V+u_1v_2+\xi_1\eta_2=Y+X+y_1x_2+\psi_1\chi_2,& \\
&u_2v_1=y_2x_1,\quad u_2\eta_1=y_2\chi_1,\quad \xi_2 v_1=\psi_2 x_1, \quad \xi_2\eta_1=\psi_2\chi_1,&\\
&Uv_1=Yx_1, \quad u_2 V=y_2 X, \quad U\eta_1=Y\chi_1, \quad \eta_2 X=\psi_2 X,& \\
&u_i+v_i=x_i+y_i, \qquad \xi_i+\eta_i=\chi_i+\psi_1, \quad i=1,2.&
\end{eqnarray}
\end{subequations}
As in the previous section, the algebraic variety consists of two components. The first ten-dimensional component corresponds to the permutation map
\begin{equation}
\textbf{x} \mapsto \textbf{u}=\textbf{y}, \qquad \textbf{y} \mapsto \textbf{v}=\textbf{x},
\end{equation}
and the second corresponds to the following ten-dimensional YB map

\begin{subequations}\label{Z25D}
\begin{align}
& \textbf{x}\, \mapsto \, \textbf{u}=\left(y_1+\frac{f(\textbf{x},\textbf{y})}{g(\textbf{x},\textbf{y})}x_1,\,\,\frac{g(\textbf{x},\textbf{y})}{h(\textbf{x},\textbf{y})}y_2,\,\,\psi_1+\frac{f(\textbf{x},\textbf{y})}{g(\textbf{x},\textbf{y})}\chi_1,\,\,\frac{g(\textbf{x},\textbf{y})}{h(\textbf{x},\textbf{y})}\psi_2,\,\,\frac{g(\textbf{x},\textbf{y})}{h(\textbf{x},\textbf{y})}Y\right), \\
& \textbf{y}\, \mapsto \, \textbf{v}=\left(\frac{h(\textbf{x},\textbf{y})}{g(\textbf{x},\textbf{y})}x_1,\,\,x_2+\frac{f(\textbf{y},\textbf{x})}{h(\textbf{x},\textbf{y})}y_2,\,\,\frac{h(\textbf{x},\textbf{y})}{g(\textbf{x},\textbf{y})}\chi_1,\,\,\chi_2+\frac{f(\textbf{y},\textbf{x})}{h(\textbf{x},\textbf{y})}\psi_2,\,\,\frac{h(\textbf{x},\textbf{y})}{g(\textbf{x},\textbf{y})}X \right).
\end{align}
\end{subequations}

where $f$, $g$ and $h$ are given by the following expressions
\begin{subequations}\label{f-g-h-Z2}
\begin{align}
\label{f-g-h-Z2-1}
&f(\textbf{x},\textbf{y})=X-x_1x_2-\chi_1\chi_2-Y+y_1y_2+\psi_1\psi_2,\\
\label{f-g-h-Z2-2}
&g(\textbf{x},\textbf{y})=X-x_1(x_2+y_2)-\chi_1(\chi_2+\psi_2),\\
\label{f-g-h-Z2-3}
&h(\textbf{x},\textbf{y})=Y-(x_1+y_1)y_2-(\chi_1+\psi_1)\psi_2.
\end{align}
\end{subequations}

\subsection{Restriction on invariant leaves}
In this section, we show that the map given by (\ref{Z25D})-(\ref{f-g-h-Z2}) can be restricted to a completely integrable eight-dimensional YB map on invariant leaves. As in the previous section, the idea of this restriction is motivated by the existence of the first integral (\ref{G-z2-1stInt}).

Particularly, we have the following.
\begin{proposition}
\begin{enumerate}
	\item $\Phi =X-x_1x_2-\chi_1\chi_2$ and $\Psi=Y-y_1y_2-\psi_1\psi_2$ are invariants of the map (\ref{Z25D})-(\ref{f-g-h-Z2}).
	\item The ten-dimensional map (\ref{Z25D})-(\ref{f-g-h-Z2}) can be restricted to an eight-dimensional map, $Y_{a,b}\in\End\{A_a\times A_b\}$, given by
	\begin{subequations}\label{Z2-Grassmann}
\begin{align}
 x_1\, \mapsto \, u_1&=y_1+\frac{(a-b)(a-x_1y_2+\chi_1\psi_2)}{(a-x_1 y_2)^2}x_1, \\
 x_2\, \mapsto \, u_2&=\frac{(a-x_1y_2-\chi_1\psi_2)(b-x_1y_2+\chi_1\psi_2)}{(b-x_1y_2)^2}y_2, \\
 \chi_1\, \mapsto \, \xi_1&=\psi_1+\frac{a-b}{a-x_1y_2}\chi_1, \\
 \chi_2\, \mapsto \, \xi_2&=\frac{a-x_1y_2}{b-x_1y_2}\psi_2, \\
 y_1\, \mapsto \, v_1&=\frac{(b-x_1y_2-\chi_1\psi_2)(a-x_1y_2+\chi_1\psi_2)}{(a-x_1y_2)^2}x_1, \\
 y_2\, \mapsto \, v_2&=x_2+\frac{(b-a)(b-x_1y_2+\chi_1\psi_2)}{(b-x_1y_2)^2}y_2, \\
 \psi_1\, \mapsto \, \eta_1&=\frac{b-x_1y_2}{a-x_1y_2}\chi_1, \\
 \psi_2\, \mapsto \, \eta_2&=\chi_2+\frac{b-a}{b-x_1y_2}\psi_2,
\end{align}
\end{subequations}	
where $a,b\in G_0$ and $A_a$, $A_b$ are given by (\ref{symleaves}).
\item The bosonic limit of the above map is the four-dimensional YB map associated to the DNLS equation.
\end{enumerate}
\end{proposition}

\begin{proof}
\begin{enumerate}
	\item Map (\ref{Z25D})-(\ref{f-g-h-Z2}) implies $U-u_1u_2-\xi_1\xi_2=X-x_1x_2-\chi_1\chi_2$ and $V-v_1v_2-\eta_1\eta_2=Y-y_1y_2-\psi_1\psi_2$. Therefore, $\Phi$ and $\Psi$ are first integrals of the map.
	\item The conditions $X=x_1x_2+\chi_1\chi_2+a$ and $Y=y_1y_2+\psi_1\psi_2+b$ define the level sets, $A_a$ and $A_b$, of $\Phi$ and $\Psi$, respectively. Using these conditions, we can eliminate $X$ and $Y$ from map (\ref{Z25D})-(\ref{f-g-h-Z2}). The resulting map, $Y_{a,b}:A_a\times A_b \longrightarrow A_a\times A_b$, is given by (\ref{Z2-Grassmann}).
  \item Setting the odd variables in (\ref{Z2-Grassmann}) equal to zero, we obtain map (4.37) in \cite{Sokor-Sasha}.	
\end{enumerate}
\end{proof}

Now, using condition $X = x_1x_2+\chi_1\chi_2+a$, matrix (\ref{DarbouxZ2}) takes the following form
\begin{equation} \label{DarbouxLaxZ2}
M=\lambda^2 \left(
\begin{matrix}
 k+x_1x_2+\chi_1\chi_2 & 0 & 0\\
 0 & 0 & 0\\
 0 & 0 & 0
\end{matrix}\right)+\lambda \left(
\begin{matrix}
 0 & x_1 & \chi_1 \\
 x_2 & 0 &0 \\
 \chi_2 & 0 & 0
\end{matrix}\right)
+\left(
\begin{matrix}
 1 & 0 & 0 \\
 0 & 1 &0 \\
 0 & 0 & 1
\end{matrix}\right).
\end{equation}
Map (\ref{Z2-Grassmann}) has the following strong Lax representation
\begin{equation} \label{lax_eq_Z2}
  M(\textbf{u};a,\lambda)M(\textbf{v};b,\lambda)=M(\textbf{y};b,\lambda)M(\textbf{x};a,\lambda).
\end{equation}
Therefore, it is reversible parametric YB map which is birational due to Prop. \ref{rationality}. It can also be verified that it is not involutive.

Regarding the invariants of map (\ref{Z2-Grassmann}), the ones which we retrieve from $\str(M(\textbf{y};b,\lambda)M(\textbf{x};a,\lambda))$ are
\begin{eqnarray*}
&&K_1=(a+x_1x_2+\chi_1\chi_2)(b+y_1y_2+\psi_1\psi_2) \\
&&K_2=x_1x_2+y_1y_2+x_1y_2+x_2y_1+\chi_1\chi_2+\psi_1\psi_2+\chi_1\psi_2-\chi_2\psi_1,
\end{eqnarray*}
where we have omitted the additive constants. However, $K_1$ is sum of the following quantities
\begin{subequations}
\begin{align}
I_1&=(a+x_1x_2)(y_1y_2+\psi_1\psi_2)+b(x_1x_2+\chi_1\chi_2)+y_1y_2\chi_1\chi_2\\
I_2&=\chi_1\chi_2\psi_1\psi_2,
\end{align}
\end{subequations}
which are invariants themselves. Moreover, $K_2$ is sum of the following invariants
\begin{equation}
I_3=(x_1+y_1)(x_2+y_2) \qquad \text{and} \qquad I_4=(\chi_1+\psi_1)(\chi_2+\psi_2).
\end{equation}
In fact, the quantities $C_i=x_i+y_i$ and $\Omega_i=\chi_i+\psi_i$, $i=1,2$, are invariants themselves.

\subsection{Vector generalisations: $4N\times4N$ maps}
In what follows we use the following notation for a vector $\textbf{w}=(w_1,...,w_{4n})$
\begin{eqnarray}
&& \textbf{w}=(\textbf{w}_1,\textbf{w}_2,\textbf{$\boldsymbol\omega$}_1,\textbf{$\boldsymbol\omega$}_2),\qquad \text{where} \qquad \textbf{w}_1=(w_1,...,w_N), \qquad \textbf{w}_2=(w_{N+1},...,w_{2N})\nonumber \\
&& \text{and} \qquad \textbf{$\boldsymbol\omega$}_1=(w_{2N+1},...,w_{3N}), \qquad \textbf{$\boldsymbol\omega$}_2=(w_{3N+1},...,w_{4N}), \nonumber
\end{eqnarray}
where $\textbf{w}_1$ and $\textbf{w}_2$ are even and $\textbf{$\boldsymbol\omega$}_1$ and $\textbf{$\boldsymbol\omega$}_2$ are odds.
Also,
\begin{equation}
\langle u_i|:=\textbf{u}_i,\qquad |w_i\rangle:=\textbf{w}_i^T \qquad \text{and their dot product with}\qquad \langle u_i,w_i\rangle.
\end{equation}

\subsection{NLS case}
Now, we replace the variables in map (\ref{Adler-Yamilov-Grassmann}) with $N-$vectors, namely we consider the following $4N\times 4N$ map
\begin{subequations}\label{VectorNLS-G}
\begin{align}
\begin{cases}
\langle u_1|=\langle y_1|+f(z;a,b)\langle x_1|(1+\langle x_1,y_2\rangle-\langle \chi_1,\psi_2\rangle),\\
\langle u_2|=\langle y_2|, \\
\langle \xi_1|=\langle \psi_1|+f(z;a,b)\langle \chi_1|(1+\langle x_1,y_2\rangle-\langle \chi_1,\psi_2\rangle), \\
\langle \xi_2|=\langle \psi_2|,
\end{cases}
\\
\intertext{and} \nonumber
\\
\begin{cases}
\langle v_1|=\langle x_1|,\\
\langle v_2|=\langle x_2|+f(z;b,a)\langle y_2|(1+\langle x_1,y_2\rangle-\langle \chi_1,\psi_2\rangle),\\
\langle \eta_1|=\langle \chi_1| \\
\langle \eta_2|=\langle \chi_2|+f(z;b,a)\langle \psi_2|(1+\langle x_1,y_2\rangle-\langle \chi_1,\psi_2\rangle)
\end{cases}
\end{align}
\end{subequations}
where $f$ is given by
\begin{equation}\label{f-g-h-VNLS}
f(z;b,a)=\frac{b-a}{(1+z)^2},\qquad z:=\langle x_1,y_2\rangle.
\end{equation}

Map (\ref{VectorNLS-G})-(\ref{f-g-h-VNLS}) is a reversible parametric YB map, for it has the following strong Lax-representation
\begin{equation}\label{LaxRep}
M(\textbf{u};a)M(\textbf{v};b)=M(\textbf{y};b)M(\textbf{x};a)
\end{equation}
where
\begin{equation}
M(\textbf{w};a)=\left(
\begin{matrix}
\lambda+a+\langle w_1,w_2\rangle+\langle \omega_1,\omega_2\rangle  & \langle w_1| & \langle \omega_1| \\
|w_2\rangle        \\
|\omega_2\rangle      & \quad \qquad \qquad I_{2N-1} &
\end{matrix}\right).
\end{equation}
Moreover, map (\ref{VectorNLS-G})-(\ref{f-g-h-VNLS}) is birational and not involutive.

The invariants of this map are given by
\begin{subequations}
\begin{align}
 K_1&=a+b+\langle x_1,x_2\rangle+\langle y_1,y_2\rangle+\langle \chi_1,\chi_2\rangle+\langle \psi_1,\psi_2\rangle, \\
 K_2&=b(\langle x_1,x_2\rangle+\langle \chi_1,\chi_2\rangle)+a(\langle y_1,y_2\rangle+\langle \psi_1,\psi_2\rangle)+(\langle y_1,y_2\rangle+\langle \psi_1,\psi_2\rangle)(\langle x_1,x_2\rangle+\langle \chi_1,\chi_2\rangle).
\end{align}
\end{subequations}
The quantities
\begin{equation}
I_1=\langle x_1,x_2\rangle+\langle \chi_1,\chi_2\rangle,\quad\text{and}\quad I_2=\langle y_1,y_2\rangle+\langle \psi_1,\psi_2\rangle,
\end{equation}
are invariant themselves, while $K_2$ can be written as sum of the following
\begin{subequations}
\begin{align}
I_3=&\langle \chi_1,\chi_2\rangle \langle \psi_1,\psi_2\rangle, \\
I_4=&b(\langle x_1,x_2\rangle+\langle \chi_1,\chi_2\rangle)+a(\langle y_1,y_2\rangle+\langle \psi_1,\psi_2\rangle)+\\
&(\langle y_1,y_2\rangle+\langle \psi_1,\psi_2\rangle)\langle x_1,x_2\rangle+\langle y_1,y_2\rangle\langle \chi_1,\chi_2\rangle.\nonumber
\end{align}
\end{subequations}

\subsection{DNLS case}
Now, replacing the variables in (\ref{Z2-Grassmann}) with $N$-vectors we obtain the following $4N\times 4N$-dimensional map
\begin{subequations}\label{vectorZ2-G}
\begin{align}
\begin{cases}
\langle u_1|&=\langle y_1|+h(z;a,b)\langle x_1|(a-\langle x_1,y_2\rangle+\langle \chi_1,\psi_2\rangle), \\
\langle u_2|&=g(z;a,b)\langle y_2|-h(z;b,a)\langle \chi_1,\psi_2\rangle \langle y_2|,\\
\langle \xi_1|&=\langle \psi_1|+f(z;a,b)\langle \chi_1|,\\
\langle \xi_2|&=g(z;a,b)\langle \psi_2|,
\end{cases}
\\
\intertext{and}\nonumber
\begin{cases}
\langle v_1|&=g(z;b,a)\langle x_1|-h(z;a,b)\langle x_1|,\\
\langle v_2|&=\langle x_2|+h(z;b,a)\langle y_2|(b-\langle x_1,y_2\rangle+\langle \chi_1,\psi_2\rangle),\\
\langle \eta_1|&=g(z;b,a)\langle \chi_1|,\\
\langle \eta_2|&=\langle \chi_2|+f(z;b,a)\langle \psi_2|,
\end{cases}
\end{align}
\end{subequations}
where $f$, $g$ and $h$ are given by
\begin{equation}\label{f-g-h-VZ2}
f(z;a,b)=\frac{a-b}{a-z},\qquad g(z;a,b)=\frac{a-z}{b-z},\qquad h(z;a,b)=\frac{a-b}{(a-z)^2}, \qquad z:=<x_1,y_2>.
\end{equation}

Map (\ref{vectorZ2-G})-(\ref{f-g-h-VZ2}) is reversible parametric YB map, as it has the strong Lax-representation (\ref{LaxRep}) where
\begin{equation} \label{DarbouxVectorLaxZ2}
M= \left(
\begin{matrix}
 \lambda^2(k+\langle x_1,x_2\rangle+\langle \chi_1,\chi_2\rangle) & \lambda \langle x_1|  & \lambda \langle \chi_1|\\
 \lambda |x_2\rangle &  & \\
 \lambda |\chi_2\rangle & \qquad \quad I_{2N} &
\end{matrix}\right).
\end{equation}
Moreover, it is a non-involutive map and birational.

The invariants we retrieve from the supertrace\index{supertrace} of the monodromy matrix\index{monodromy matrix} are given by
\begin{eqnarray*}
&&K_1=(a+\langle x_1,x_2\rangle+\langle\chi_1,\chi_2\rangle)(b+\langle y_1,y_2\rangle+\langle\psi_1,\psi_2\rangle) \\
&&K_2=\langle x_1,x_2\rangle+\langle y_1,y_2\rangle+\langle x_1,y_2\rangle+\langle x_2,y_1\rangle+\langle\chi_1,\chi_2\rangle+\langle\psi_1,\psi_2\rangle+\\
&&\qquad~~\langle\chi_1,\psi_2\rangle-\langle\chi_2\psi_1\rangle,
\end{eqnarray*}
where we have omitted the additive constants. In fact, $K_2$ is a sum of the following invariants
\begin{equation}
I_1=\langle x_1+y_1,x_2+y_2\rangle,\qquad I_2=\langle\chi_1+\psi_1,\chi_2+\psi_2 \rangle,
\end{equation}
and the vectors in the above dot products are invariant themselves, namely
\begin{equation}
\langle x_1+y_1|,\qquad \langle x_2+y_2|,\qquad \langle\chi_1+\psi_1|,\qquad \text{and} \qquad \langle \chi_2+\psi_2|,
\end{equation}
are invariants.

\section{Conclusions}
We showed that there are explicit examples of birational endomorphisms of Grassmann algebraic varieties which possess the Yang-Baxter property. These YB maps are related to noncommutative versions of integrable PDEs via their Lax representations which consist of Darboux matrices for these PDEs. Specifically, we considered the cases of the Grassmann extensions of Darboux transformations corresponding to
\begin{enumerate}
	\item the NLS equation;
	\item the DNLS equation.
\end{enumerate}
In the former case a Darboux transformation appeared in \cite{Georgi-Sasha} and, here, we constructed a Darboux transformation for the latter case. Employing the associated Darboux matrices we derived ten-dimensional maps, which we restricted on invariant leaves to eight-dimensional birational parametric YB maps. The motivation for these restrictions was the fact that the entries of the associated Darboux matrices satisfy particular systems of differential-difference equations which possess first integrals. The latter indicated the invariant leaves. In the case of the NLS equation the derived eight-dimensional YB map, namely map (\ref{Adler-Yamilov-Grassmann}), is the Grassmann extension of the Adler-Yamilov map, while in the case of the DNLS equation the result is a novel eight-dimensional YB map, map (\ref{Z2-Grassmann}), which, at the bosonic limit, is equivalent to a four-dimensional YB map which appeared recently in \cite{Sokor-Sasha}. Moreover, we considered the vector generalisations of these eight-dimensional maps.

Our results could be extended in several ways.
\begin{enumerate}[i]
  \item find the Poisson structure of the eight-dimensional YB maps;
	\item study the case of the Lax operator with $\field{D}_2$ symmetries;
  \item study the corresponding noncommutative entwining systems;
	\item study the transfer dynamics of all the Grassmann extended YB maps and the entwining systems associated to the Grassmann extended Darboux matrices.
\end{enumerate}
Regarding (i) one needs to find the even-odd Poisson brackets \cite{Berezin} for the maps (\ref{Adler-Yamilov-Grassmann}) and (\ref{Z2-Grassmann}). In the case of map (\ref{Z2-Grassmann}), the invariants $C_i$ and $\Omega_i$ will be Casimirs for the associated Poisson bracket. For (ii), in \cite{SPS} we studied the Darboux transformations in the case of Lax operators which are invariant under the action of the $\field{D}_2$ reduction group, whereas in \cite{Sokor-Sasha} we studied the associated YB maps. The Grassmann extension of these Darboux transformations and their associated YB maps in this case is an open problem. With regards to (iii), one can consider Lax triples of Darboux matrices with even and odd entries. Finally, concerning (iv), one can consider the transfer maps for the $n$-periodic problem as defined in \cite{KouloukasBanach}.

\section*{Acknowledgements}
The authors would like to thank Vassilis Papageorgiou for the private discussion, and Hovhannes Khudaverdian for the discussion and his useful comments. G.G.G and A.V.M. acknowledge support from Leverhulme Trust. The work of A.V.M. is partially supported by the EPSRC (Grant EP/I038675/1 is acknowledged). S.K.R acknowledges University of Leeds' William Wright Smith scholarship and would like to thank J.E. Crowther for the scholarship-contribution to fees. The work of S.K.R was finalised at the ``Laboratory of Applied Mathematics \& Computer Technology" of the Faculty of Applied Mathematics \& Computer Technology, Chechen State University, Russia.


\begin{thebibliography}{10}


\bibitem{ABS-2004}
{\sc Adler V, Bobenko A, and Suris Y} {2003} {Classification of
  integrable equations on quad-graphs. {T}he consistency approach} {\em Comm. Math.
  Phys.} {\textbf{233}} {513--543}.

\bibitem{ABS-2005}
{\sc Adler V, Bobenko A, and Suris Y} {2004} { Geometry of
  {Y}ang-{B}axter maps: pencils of conics and quadrirational mappings} {\em Comm.
  Anal. Geom.} {\textbf{12}} {967--1007}.


\bibitem{Adler-Yamilov}
{\sc Adler V and Yamilov R} {1994} {Explicit auto-transformations of
  integrable chains} {\em J. Phys. A} {\textbf{27}} {477--492}.

\bibitem{Berezin}
{\sc Berezin F.} {1987} {Intoduction to superanalysis}
  {\em Reidel Publishing Co. Dordrecht}  {\textbf{9}}.

\bibitem{Bobenko-Suris}
{\sc Bobenko A and Suris Y.} {2002} {Integrable systems on quad-graphs}
  {\em Int. Math. Res. Not.}  {573--611}.

\bibitem{Buchstaber}
{\sc Buchstaber V} {1998} {The Yang-Baxter transformation}
  {\em Russ. Math. Surveys} {\textbf{53:6}} {1343-1345}.


\bibitem{BuryPhD}
{\sc Bury R} {2010} {Automorphic {L}ie algebras, corresponding integrable
  systems and their soliton solutions} {\em Ph.D. thesis, Un. of Leeds}.

\bibitem{Bury-Sasha}
{\sc Bury R and Mikhailov A} {2012} {Automorphic {L}ie algebras and
  corresponding integrable systems. {I}.} (to be submited)


\bibitem{Chain-Kulish}
{\sc Chaichian M. and Kulish P.} {1978} {On the method of inverse scattering problem and B\"acklund transformations for supersymmetric equations} {\em Phys. Lett. B} {\textbf{78}} {413--416}.


\bibitem{Dimakis}
{\sc Dimakis A. and M\"uller-Hoissen F.} {2005} {An algebraic scheme associated with the non-commutative {KP} hierarchy and some of its extensions} {\em J. Phys. A} {\textbf{38}} {5453--5505}.


\bibitem{Doliwa}
{\sc Doliwa A} {2014} {Non-commutative rational Yang-Baxter maps} {\em Lett. Math. Phys} {\textbf{104}} {299--309}.


\bibitem{Drinfel'd}
{\sc Drinfeld V} {1992} {On some unsolved problems in quantum group
  theory} {\em Lecture Notes in
  Math.} {\textbf{1510}} {1--8}.


\bibitem{Georgi-Sasha}
{\sc Grahovski G and Mikhailov A} {2013} {Integrable discretisations for a class of nonlinear Scr\"odinger equations on {G}rassmann algebras} {\em Phys. Lett. A} {\textbf{377}} {3254--3259}.


\bibitem{Kaup-Newell}
{\sc Kaup D and Newell A} {1978} {An exact solution for a derivative
  nonlinear {S}chr\"odinger equation} {\em J. Mathematical Phys.} {\textbf{19}} {798--801}.

\bibitem{Sokor-Sasha}
{\sc Konstantinou-Rizos S and Mikhailov A} {2013} {Darboux transformations, finite reduction groups and related Yang-Baxter maps} {\em J. Phys. A} {\textbf{46}} {425201}.


\bibitem{SPS}
{\sc Konstantinou-Rizos S and Mikhailov A and Xenitidis P} {2015} {Reduction groups and related integrable difference systems of nonlinear {S}chr\"odinger type} {\em J. Math. Phys.} {\textbf{56}} {082701}.

\bibitem{Kouloukas}
{\sc Kouloukas T and Papageorgiou V} {2009} {Yang-{B}axter maps with
  first-degree-polynomial {$2\times 2$} {L}ax matrices} {\em J. Phys. A} {\textbf{42}} {404012}.


\bibitem{KouloukasBanach}
{\sc Kouloukas T and Papageorgiou V} {2011} {Entwining {Y}ang-{B}axter maps and integrable
lattices} {\em Banach Centre Publications} {\textbf{93}} { 163--175}.

\bibitem{Kouloukas2}
{\sc Kouloukas T and Papageorgiou V} {2011} {Poisson Yang-Baxter
  maps with binomial Lax matrices} {\em J. Math. Phys.} {\textbf{52}} {404012}.
	
	



\bibitem{Lombardo}
{\sc Lombardo S and Mikhailov A} {2005} {Reduction groups and automorphic
  {L}ie algebras} {\em Comm. Math. Phys.} {\textbf{258}} {179--202}.



\bibitem{Frank4}
{\sc Nijhoff F and Walker A} {2001} {The discrete and continuous
  {P}ainlev\'e {VI} hierarchy and the {G}arnier systems} {\em Glasg. Math. J.} {\textbf{43A}}
  {109--123}.



\bibitem{PT}{\sc Papageorgiou V and Tongas A} {2007} {Yang-{B}axter
  maps and multi-field integrable lattice equations} {\em J. Phys. A}
  {\textbf{40}} {083502}.


\bibitem{Sklyanin}
{\sc Sklyanin E} {1988} {Classical limits of $SU(2)$-invariant solutions of the {Y}ang-{B}axter equation},
  {\em J. Soviet Math.} {\textbf{40}}(1) {93--107}.
		

\bibitem{Veselov2}
{\sc Suris Y and Veselov A} {2003} {Lax matrices for {Y}ang-{B}axter
  maps} {\em J. Nonlinear Math. Phys.} {\textbf{10}} {223--230}.


\bibitem{Veselov}
{\sc Veselov A} {2003} {Yang-{B}axter maps and integrable dynamics} {\em Phys.
  Lett. A} {\textbf{314}} {no. 3} {214--221}.

\bibitem{Veselov3}
{\sc Veselov A} {2007} {Yang-{B}axter maps: dynamical point of view} {\em Math. Soc. Japan} {\textbf{17}} {145--167}.

\bibitem{Liu}
{\sc Xue L, Levi D and Liu Q} {2013} {Supersymmetric KdV equation: Darboux transformation and discrete systems} {\em J. Phys. A} {\textbf{46}} {502001}.

\bibitem{Liu2}
{\sc Xue L  and Liu Q} {2014} {B\"acklund-Darboux transformations and discretizations of super KdV equation} {\em SIGMA} {\textbf{10}} {paper 045 (10 pages)}.

\bibitem{Liu3}
{\sc Xue L  and Liu Q} {2015} {A supersymmetric AKNS problem and its Darboux-B\"acklund transformations and discrete systems} {\em Stud. Appl. Math.} {\textbf{135}} {35--62}.

\bibitem{ZS}
{\sc Zakharov V and Shabat A} {1979} {Integration of the nonlinear equations of mathematical physics by the method of the inverse scattering problem. {II}} {\em Funktsional. Anal. i Prilozhen.} {\textbf{13}} {13--22}.

\bibitem{Zhou}
{\sc Zhou R} {2014} {A Darboux transformation of the $sl(2|1)$ super KdV hierarchy and a super lattice potential KdV equation} {\em Phys. Lett. A} {\textbf{378}} {1816--1819}.

\end{thebibliography}
\end{document}